\documentclass[11pt, a4paper]{amsart}

\usepackage[latin1]{inputenc}
\usepackage[all]{xy}
\usepackage{amssymb, amsmath, amscd, geometry,mathtools}
\usepackage{braket}

\usepackage{graphicx}

\numberwithin{equation}{section}
\usepackage{latexsym}
\usepackage[usenames]{color}
\usepackage{epic} %%%% Circunferencias grandes y curvas de bezier
\usepackage{mathrsfs}
\usepackage{euscript}
\usepackage{bbm}
\usepackage{rotating}
\usepackage{url}

\usepackage{verbatim}
\usepackage[usenames]{color}

\theoremstyle{plain}
\newtheorem{theo}{Theorem}[section]

\newtheorem{proposition}[theo]{Proposition}
\newtheorem{theorem}[theo]{Theorem}
\newtheorem{corollary}[theo]{Corollary}
\newtheorem{lemma}[theo]{Lemma}
\theoremstyle{definition}

\newcommand{\K}{{\mathbb K}}      \newcommand{\N}{{\mathbb N}}

\def\C{\mathbb{C}}

\newcommand{\one}{\mbox{$1 \hspace{-1.0mm}  {\bf l}$}}

\begin{document}

% \numberwithin{equation}{section}

%\pagestyle{headings}

\author{}

\title{Unbounded Bell violations for quantum genuine multipartite non-locality}

\author{Abderram\'an Amr $^1$}
\author{Carlos Palazuelos $^{1,3}$}
\author{Julio I. de Vicente$^2$}
\address[$^1$]{Departamento de An\'alisis Matem\'atico y Matem\'atica Aplicada, Universidad Complutense de Madrid, Plaza de Ciencias 3, 28040 Madrid, Spain}
\address[$^2$]{Departamento de Matem\'aticas, Universidad Carlos III de Madrid, Avda. de la Universidad 30, 28911 Legan\'es, Madrid, Spain}
\address[$^3$]{Instituto de Ciencias Matem\'aticas (ICMAT), C/ Nicol\'as Cabrera, Campus de Cantoblanco, 28049 Madrid, Spain}
\email[$^1$]{abamr@ucm.es}
\email[$^2$]{carlospalazuelos@mat.ucm.es}
\email[$^3$]{jdvicent@math.uc3m.es}

\maketitle

\begin{abstract}
The violations of Bell inequalities by measurements on quantum states give rise to the phenomenon of quantum non-locality and express the advantage of using quantum resources over classical ones for certain information-theoretic tasks. The relative degree of quantum violations has been well studied in the bipartite scenario and in the multipartite scenario with respect to fully local behaviours. However, the multipartite setting entails a more complex classification in which different notions on non-locality can be established. In particular, genuine multipartite non-local distributions apprehend truly multipartite effects, given that these behaviours cannot be reproduced by bilocal models that allow correlations among strict subsets of the parties beyond a local common cause. We show here that, while in the so-called correlation scenario the relative violation of bilocal Bell inequalities by quantum resources is bounded, i.e. it does not grow arbitrarily with the number of inputs, it turns out to be unbounded in the general case. We identify Bell functionals that take the form of non-local games for which the ratio of the quantum and bilocal values grows unboundedly as a function of the number of inputs and outputs.
\end{abstract}

\section{Introduction}

Bell's theorem establishes that correlations among the results of spatially separated measurements on composite quantum systems are incompatible  with  a  local  variable model \cite {Bell}. Besides its crucial implications in the foundations of quantum mechanics, this  phenomenon, known  as \emph{quantum non-locality}, is behind many important applications in quantum information theory, such as quantum cryptography \cite{Ekert91,ABGMPS,Vazirani14}, communication complexity \cite{CleveBuhrman97}, randomness expansion \cite{Pironio10} and randomness amplification \cite{Colbeck12,Gallego13}. The simplest scenario that enables quantum non-locality considers two isolated parties and has been thoroughly studied over the last three decades (see e.g.\ the review \cite{BCPSW}). However, multipartite scenarios have a greater complexity and offer a much richer source of correlations. The potential applications of these phenomena in the context of quantum networks or many-body physics has triggered notable interest in the study of quantum non-locality in the multipartite setting in the last years.

The most natural extension of the notion of locality to the multipartite domain is that of full locality, in which the only source of correlations among the parties is a local common cause. However, the verification of non-fully-local correlations does not necessarily imply non-locality shared among all parties as this resource distributed among just two parties can be enough to falsify these models. Thus, as it also happens in the study of quantum entanglement, a notion of genuine multipartite non-locality (GMNL) can be established in order to capture the idea that the non-local correlations must be truly shared among all parties. In particular, in his celebrated paper \cite{sve}, Svetlichny proved that there exist tripartite quantum correlations that cannot be reproduced by a larger class of models in which, in addition to the local variable, arbitrary (even signalling) correlations are allowed within a strict subset of the parties. This idea, which was extended to an arbitrary number of parties in \cite{Collins02,svegen}, leads to the definition of what we will call general bilocal models. The impossibility of such a model to reproduce certain correlations thus enables the verification of GMNL. Interestingly, this introduces a rich theoretical structure because in these hybrid models one can impose different conditions to the kind of correlations that are allowed within the subsets of parties \cite{Jones05,Bancal09}. In particular, as observed in \cite{Gallego12,Bancal13}, in certain scenarios allowing for signalling correlations in the definition of bilocality leads to ill-defined resource theories and grandfather-style paradoxes. A natural way to get rid of these problems is to consider bilocal models in which correlations among subsets of parties are required to be non-signalling (see e.g.\ \cite{Almeida10,Curchod19}). We will refer to these models as non-signalling bilocal. They introduce an alternative weaker notion of GMNL as non-signalling bilocal models are general bilocal but not necessarily the other way around.

One of the main aims in the study of non-locality from an information-theoretical perspective is to identify and quantify possible advantages in the use of quantum non-local correlations over local ones. Most of these potential protocols boil down mathematically to linear functionals acting on the space of joint conditional probability distributions for the parties and, hence, this leads to the investigation of quantum violations of Bell inequalities. Thus, a natural and often used quantity to understand the difference between local and quantum non-local behaviours in a resource-theoretic way is the relative ratio of violation optimized over all possible Bell inequalities, a question that is moreover related to classical problems in functional analysis. Motivated by the seminal work of Tsirelson \cite{Tsirelson}, a series of papers have been devoted to studying the asymptotic behaviour of this ratio and have shown that this study is very suitable to tackle different problems involving quantum non-locality \cite{BuhrmanRSW12, JungeP11low, JPPVW}. In particular, this finds immediate application to dimension witnesses, communication complexity and entangled games \cite{JPPVW2, PV-Survey}. In this context, the classical result of Tsirelson \cite{Tsirelson}, which states that the above quantity is upper bounded by a universal constant (Grothendieck's constant) when one considers two-output correlations independently of the number of inputs and the Hilbert space dimension, can be understood as a limitation of the advantages of quantum mechanics. This motivated the question, posed by Tsirelson himself, of whether a similar result was true for tripartite correlations and which was answered in \cite{PerezWPVJ08tripartite} in the negative comparing non-local and fully local behaviours. In this sense, tripartite quantum correlations lead to unbounded Bell violations over fully local ones, something that suggests, at least on the theoretical level, the idea of ``unlimited advantage''. Beyond the two-output correlation scenario, it is also known that Bell violations are unbounded for general bipartite probability distributions (see e.g.\ \cite{BuhrmanRSW12,JungeP11low}).

According to the more general notions of non-locality described above that arise in the multipartite scenario, one can wonder whether the main result in \cite{PerezWPVJ08tripartite} is still true in this context. Thus, the goal of this paper is to study the relative violation of Bell inequalities for quantum multipartite behaviours with respect to bilocal ones. Our first result is that in the two-output correlation scenario, contrary to the full-locality case and on the analogy of Tsirelson's result for the bipartite case, there is a universal constant which prevents from having unbounded violations irrespectively of the notion of bilocality used (general bilocal models boil down to non-signalling bilocal models in this case). Next, we consider the same question for general probability distributions. Our main result is that in this case quantum GMNL systems lead to unbounded Bell violations with respect to general bilocal models and, hence, with respect to all other bilocal models since this is the strongest notion of bilocality. In more detail, we provide an instance of a Bell functional, which happens to be a three-prover one-round game, for which the ratio of the quantum value over bilocal models grows with the number of inputs and ouputs. Interestingly, for this it is already enough to consider tripartite systems and, hence, the result extends trivially to an arbitrary number of parties. Thus, although most of our claims are easily generalized to general multipartite systems, we stick throughout the text to the tripartite case, which has the additional benefit of alleviating considerably the notation. The techniques we use are elementary and should be accessible to all readers familiar with quantum non-locality. It turns out that the constraints of the bipartite case can be naturally extended to study relative Bell violations in the genuine multipartite setting. In particular, the aforementioned games that we consider are built as tensor products of bipartite games in the flavour of parallel repetition.

\section{General definitions}

We will consider here the standard Bell scenario in which $k$ parties produce outputs $\{a_i\}_{i=1}^k$ upon receiving inputs $\{x_i\}_{i=1}^k$ according to the joint conditional probability distribution, also called behaviour,
\begin{equation}\label{behaviour}
(P(a_1,\cdots,a_k|x_1,\cdots , x_k))_{x_1,\ldots, x_k}^{a_1, \ldots, a_k}.
\end{equation}

For simplicity, we will assume that both the input and output alphabets have the same cardinality for all parties. We will denote by $N$ the number of possible inputs and by $K$ the number of possible outputs. A distribution (\ref{behaviour}) is said to be fully local if
\begin{align}\label{Def NL Tri}
P(a_1,\cdots,a_k|x_1,\cdots , x_k)=\sum_\lambda p_\lambda  P_1(a_1|x_1,\lambda)\cdots P_k(a_k|x_k,\lambda),
\end{align}
where $(p_\lambda)_\lambda$ denotes a probability distribution and $(P_i(a_i|x_i,\lambda))_{x_i,a_i}$ is a conditional probability distribution in the $i^{th}$ party for every value of the local variable $\lambda$. We will denote by $\mathcal{L}^k(N,K)$  the set of $k$-partite fully local probability distributions for $N$ inputs and $K$ outputs per party. On the other hand, bilocal behaviours admit a more general model of the form
\begin{align}\label{Def Sve}
&P(a_1,\cdots,a_k|x_1,\cdots , x_k)\\\nonumber&=\sum_{M,\lambda}p_{M,\lambda}P_M((a_i)_{i\in M}|(x_i)_{i\in M},\lambda)P_{\bar{M}}((a_i)_{i\in \bar{M}}|(x_i)_{i\in \bar{M}},\lambda),
\end{align}
where $M$ runs over all strict non-empty subsets of $\{1,\cdots, k\}$, $\bar{M}=\{1,\cdots, k\}\setminus M$, $(p_{M,\lambda})_\lambda$ is a probability distribution $\forall M$ and for each $M$, $P_M((a_i)_{i\in M}|(x_i)_{i\in M},\lambda)$ and $P_{\bar{M}}((a_i)_{i\in \bar{M}}|(x_i)_{i\in \bar{M}},\lambda)$ are conditional probability distributions on the parties in $M$ and in $\bar{M}$ respectively for all values of the local variable $\lambda$. If no restriction is added to these conditional probability distributions for the subsets of parties, we have Svetlichny's notion of bilocality. We will denote the set of such behaviours by $\mathcal{BL}^k_{\mathcal{G}}(N,K)$ and we will refer to them as general bilocal. If, on the other hand, the conditional probability distributions for the subsets are required to be non-signalling, we will refer to these behaviours as non-signalling bilocal and we will denote the corresponding set by $\mathcal{BL}^k_{\mathcal{NS}}(N,K)$. We recall that non-signalling conditional probability distributions are such that each party's marginal conditional probability distribution is independent of the other parties' inputs. That is, a bipartite conditional probability distribution $P(a,b|x,y))_{x,y}^{a,b}$ is said to be non-signalling if it verifies
\begin{align}\label{nonsignalling}
\sum_{a=1}^KP(a,b|x,y)=\sum_{a=1}^KP(a,b|x'y)\text{ for all }a,b,y,x\neq x',\\
\nonumber \sum_{b=1}^KP(a,b|x,y)=\sum_{b=1}^KP(a,b|xy')\text{ for all }a,b,x,y\neq y'.
\end{align}

These two conditions can be generalized in the obvious way to any number of parties and we will denote by $\mathcal{NS}^k(N, K)$ the set of $k$-partite non-signalling probability distributions with $N$ inputs and $K$ outputs per party (see \cite[Definition 1]{LW} for the explicit statement). Finally, a behaviour (\ref{behaviour}) is quantum if
\begin{align}\label{Def quantum}
P(a_1,\cdots,a_k|x_1,\cdots , x_k)=\langle \psi|E_{a_1,x_1}^{1}\otimes \cdots \otimes E_{a_k,x_k}^{k}|\psi\rangle,
\end{align}
where $(E_{a_i,x_i}^{i})_{x_i, a_i}$ is a family of measurements for the i$^{th}$-party (that is, for each party $i$ $E_{a_i,x_i}^{i}\geq0$ $\forall a_i,x_i$ and $\sum_{a_i}E_{a_i,x_i}^{i}=\one$ $\forall x_i$) and $|\psi\rangle$ is a $k$-partite pure quantum state. We will denote the set of behaviours of this form by $\mathcal{Q}^k(N,K)$.

There are several known relations among these sets. For instance, $\mathcal{L}^k(N,K)\subsetneq \mathcal{Q}^k(N,K)\subsetneq \mathcal{NS}^k(N,K)$. The first strict inclusion is the content of Bell's theorem while the second is due to Tsirelson \cite{Tsirelson} and Popescu and Rohrlich \cite{PR}. It readily follows from the definitions that $\mathcal{L}^k(N,K)\subset \mathcal{BL}^k_{\mathcal{NS}}(N,K)\subset \mathcal{BL}^k_{\mathcal{G}}(N,K)$. Svetlichny's result states that $\mathcal{Q}^k(N,K)\nsubseteq \mathcal{BL}^k_{\mathcal{G}}(N,K)$. Notice that it also holds that $\mathcal{BL}^k_{\mathcal{NS}}(N,K)\nsubseteq \mathcal{Q}^k(N,K)$.

Given any linear functional $M$ acting on the set of $k$-partite joint conditional probability distributions characterized by real numbers $\{M_{x_1 \ldots x_k}^{a_1 \ldots a_k}\}$ we will write
$$\langle M,P\rangle=\sum_{a_1, \ldots, a_k}\sum_{x_1, \ldots, x_k}M_{x_1 \ldots x_k}^{a_1 \ldots a_k}P(a_1,\cdots,a_k|x_1,\cdots , x_k).$$

Many information-theoretic problems boil down to optimizing a linear functional over particular sets of behaviours and, as explained in the introduction, a good way of understanding the relative power as resources of two such sets is to consider the relative ratio of violation optimized over all linear functionals. More precisely, if $\mathcal A_1$ and $\mathcal A_2$ are certain sets of behaviours like those defined above and $M$ is a linear functional on them, we define $\omega_{\mathcal{A}_i}(M)=\sup_{P\in\mathcal{A}_i}|\langle M,P\rangle|$, $i=1,2$ and also
\begin{align}\label{Larg viol}
LV(\mathcal A_1, \mathcal A_2)=\sup_M \frac{\omega_{\mathcal{A}_1}(M)}{\omega_{\mathcal{A}_2}(M)}.
\end{align}
This quantity will be our major object of study here in order to compare $\mathcal{Q}^k(N,K)$ with the sets of bilocal behaviours. As discussed in the introduction, Eq.\ (\ref{Larg viol}) gives a quantitative notion of the relative power as a resource of the behaviours in $\mathcal{A}_1$ compared to those in $\mathcal{A}_2$. This quantification is particularly clear when the Bell inequality $M$ is a two-prover one-round game, where Eq. (\ref{Larg viol}) is exactly the quotient of the winning probability of the game by using strategies defined by $\mathcal {A}_1$ over the the winning probability of the game by using strategies defined by $\mathcal A_2$.

We note that, in order for Eq.\ (\ref{Larg viol}) to make sense, we must require that the set $\mathcal A_1$ is contained in the affine hull of the set $\mathcal{A}_2$ (otherwise, we could find examples where  $0=\omega_{\mathcal{A}_2}(M)<\omega_{\mathcal{A}_1}(M)$, so the quotient is infinity) and we must also define $0/0=0$ if we want to allow for general Bell functionals $M$ in the equation. We refer to \cite[Section 5]{JungeP11low} for a complete study of the geometric interpretation of Eq.\ (\ref{Larg viol}) in the bipartite case. On the other hand, this restriction is not needed anymore if we restrict to the case of correlations (see the following section) or to Bell functionals $M$ with positive coefficients (in particular, two-prover one-round games), since in both cases $\omega_{\mathcal A}(M)>0$ for every $M$, for all the sets $\mathcal A$ we will consider. 
\section{Correlation scenario}\label{Sec: Bi-local correlations}

A particularly simple setting that has been thoroughly studied in the literature is the so-called correlation scenario. This arises when all outputs are binary $a_i\in \{-1,1\}$ $\forall i$ (i.e.\ $K=2$) and, instead of considering the full joint conditional probability distribution (\ref{behaviour}), only correlations -- that is, expectations over the product of the outputs -- are considered. Then, we define the correlation associated to a behaviour (\ref{behaviour}), $$\gamma=(\gamma_{x_1\ldots x_k})_{x_1,\ldots, x_k}\in\mathbb{R}^{N^k},$$ as
\begin{align*}
\gamma_{x_1\ldots x_k}&=\mathbb{E}[a_1 \ldots a_k|x_1,\ldots, x_k]=\sum_{a_1, \ldots, a_k}a_1\cdots a_k P(a_1, \ldots, a_k|x_1,\ldots, x_k)\\&=P(a_1\cdots a_k{=}{1}|x_1,\ldots, x_k)-P(a_1\cdots a_k{=}{-}1|x_1,\ldots, x_k).
\end{align*}

We will say that a certain correlation is local (resp. quantum, non-signalling bilocal, general bilocal and non-signalling) if there exists a local (resp. quantum, non-signalling bilocal, general bilocal and non-signalling) joint conditional probability distribution (\ref{behaviour}) such that $\gamma$ is the correlation associated to it. In this way, we will denote, in correspondence with the definitions of the previous section, $\mathcal{L}_{cor}^k(N)$ (resp. $\mathcal{Q}_{cor}^k(N)$, $\mathcal{BL}^k_{cor, \mathcal{NS}}(N)$, $\mathcal{BL}^k_{cor, \mathcal{G}}(N)$ and $\mathcal{NS}_{cor}^k(N)$) the set of local (resp. quantum, non-signalling bilocal, general bilocal and non-signalling) correlations with $N$ inputs per party.

It is well known that a given correlation $\gamma$ is in $\mathcal{L}_{cor}^k(N)$ if
$$\gamma \in \text{conv}\{(a_{x_1}\cdots a_{x_k})_{x_1,\ldots, x_k}: a_{x_i}=\pm 1, x_i=1,\cdots, N, i=1,\cdots , k\},$$where $\text{conv}$ denotes the convex hull. On the other hand, $\gamma$ is in $\mathcal{Q}_{cor}^k(N)$ if
\begin{align*}
\gamma_{x_1\ldots x_k}=\langle \psi|A^1_{x_1}\otimes \ldots \otimes A^k_{x_k}|\psi\rangle \hspace{0.2 cm}\text{  for every $x_1,\ldots, x_k$,}
\end{align*}where $A^i_{x_i}$ is a norm-one selfadjoint operator for every $i$ and for every $x_i=1,\ldots, N$ and $|\psi\rangle$ is a $k$-partite pure quantum state.

A characterization of the set $\mathcal{NS}_{cor}^k(N)$ can be done by the following lemma (see \cite[Proposition 1,2]{communicationcomplexityns} for the proof in the case $k=2$)
\begin{lemma}\label{nscorrelation}
Given a correlation $\gamma$, it is in $\mathcal{NS}_{cor}^k(N)$ if and only if $|\gamma_{x_1\ldots x_k}|\leq 1$ for every $x_1,\ldots, x_k$.
\end{lemma}
\begin{proof}
Proving that $|\gamma_{x_1\ldots x_k}|$ is less than or equal to 1 if it is in $\mathcal{NS}_{cor}^k(N)$ follows from the definition. For the other implication, given a correlation $\gamma_{x_1,\ldots,x_k}$ satisfying $|\gamma_{x_1\ldots x_k}|\leq 1$ for every $x_1,\ldots, x_k$, consider the probability distribution defined as:
$$P(a_1,\ldots,a_k|x_1,\ldots,x_k)=\left\{
\begin{array}{ccc}
\frac{1+\gamma_{x_1\ldots x_k}}{2^k}& \text{if} & a_1a_2\ldots a_k=1,\\
\frac{1-\gamma_{x_1\ldots x_k}}{2^k}& \text{if} & a_1a_2\ldots a_k=-1.
\end{array}
\right.$$

This probability distribution can be easily seen to be in $\mathcal{NS}^k(N,2)$ as a consequence that for all $i$ we have
\begin{align*}
\sum_{a_i}&P(a_1,\ldots,a_k|x_1,\ldots,x_k)\\
&=P(a_1,\ldots,\underbrace{1}_i,\ldots,a_k|x_1,\ldots,x_k)+P(a_1,\ldots,\underbrace{-1}_i,\ldots,a_k|x_1,\ldots,x_k)\\
&=\frac{1\pm\gamma_{x_1,\ldots,x_k}}{2^k}+\frac{1\mp\gamma_{x_1,\ldots,x_k}}{2^k}=\frac{1}{2^{k-1}}.
\end{align*}
\end{proof}
In particular, notice that Lemma \ref{nscorrelation} implies that correlations associated to nonsignalling probability distributions are the same as correlations associated to general probability distributions. So we have $$\mathcal{BL}^k_{cor, \mathcal{NS}}(N)=\mathcal{BL}^k_{cor, \mathcal{G}}(N).$$ Let us then just denote $\mathcal{BL}^k_{cor}(N)$ in this case.

In the following proposition we characterize correlations associated to bilocal probability distributions. Since we will only consider the case of tripartite distributions, we will restrict to this case. However, this result generalizes to an arbitrary number of parties trivially.

\begin{proposition}\label{Prop correlation bilocal}
A correlation $(\gamma_{xyz})_{x,y,z}$ is in $\mathcal{BL}^3_{cor}(N)$ if and only if
$$\gamma\in \text{conv}\big\{(\alpha_{xy}c_z)_{x,y,z}, \, (\beta_{yz}a_x)_{x,y,z}, \, (\gamma_{xz}b_y)_{x,y,z}\big\},$$
where $(\alpha_{xy})_{x,y}$, $(\beta_{yz})_{y,z}$, $(\gamma_{xz})_{x,z}$ are elements in $\mathcal{NS}_{cor}^2(N)$ and $a_x,b_y,c_z=\pm1$ for every $x,y,z$.
\end{proposition}
\begin{proof}
An extremal probability distribution $P$ of the set $\mathcal{BL}^3_{\mathcal{G}}(N,2)$ will have one of the following forms:
\begin{align}\label{3 cases}
(Q(a,b|x,y)R(c|z))_{x,y,z}^{a,b,c},\, (Q(b,c|y,z)R(a|x)))_{x,y,z}^{a,b,c},\, (Q(a,c|x,z)R(b|y)))_{x,y,z}^{a,b,c},
\end{align}where in all cases $Q$ and $R$ are general probability distributions. Let us assume that this extremal element has the form $P=(Q(a,b|x,y)R(c|z))_{x,y,z}^{a,b,c}$ and denote $\beta$ and $\alpha$ the corresponding correlations from the probability distributions $(Q(a,b|x,y))_{x,y}^{a,b}$ and $(R(c|z))_{z}^c$, respectively. Then, given $x,y,z$, we have
\begin{align*}
\gamma_{xyz}&=\mathbb{E}[a\cdot b\cdot c|x,y,z]=\sum_{a,b,c}abc P(a,b,c|x,y,z)=\sum_{a,b,c}abc Q(a,b|x,y)R(c|z)\\
&=\Big(\sum_{a,b}ab Q(a,b|x,y)\Big)\Big(\sum_cc R(c|z)\Big)=\mathbb{E}[a\cdot b|x,y]\mathbb{E}[c|z]=\beta_{xy}\alpha_{z}.
\end{align*}

By definition, $(\beta_{xy})_{xy}$ is in $\mathcal{NS}^2_{cor}$ whenever $Q$ is in $\mathcal{NS}^2$ and, clearly, $|\mathbb{E}[c|z]|\leq 1$. Since the other two cases in (\ref{3 cases}) are completely analogous, the result follows by convexity.
\end{proof}

We are now in the position to address the main aim of this section: characterize the relative Bell violations for the quantum and bilocal sets in the correlation scenario using the figure of merit defined in Eq.\ (\ref{Larg viol}). Before presenting these results, let us first briefly recall several known results in this direction. The result of Tsirelson (\cite{Tsirelson}) mentioned in the introduction states that
\begin{align}\label{Tsirelosn}
LV\big(\mathcal{Q}_{cor}^2(N), \mathcal{L}_{cor}^2(N)\big)\leq K_G,
\end{align}for every number of inputs $N$, where $K_G$ denotes the real Grothendieck's constant. On the other hand, it is not difficult to show (see for instance \cite[Ex. 29]{LCPW}) that
\begin{align}\label{Eq. upper bound NS-CL}
LV\big(\mathcal{NS}_{cor}^2(N), \mathcal{L}_{cor}^2(N)\big)\leq \sqrt{2N},
\end{align}and that this upper bound is essentially optimal. That is, there exits $M_0\in\mathbb{R}^{N^2}$ such that $\omega_{\mathcal{NS}_{cor}^2(N)}(M_0)\geq \sqrt{N/2}\omega_{\mathcal{L}_{cor}^2(N)}(M_0)$. As mentioned in the introduction, the relative violation between the quantum and the fully local set turns out to be unbounded as shown in \cite{PerezWPVJ08tripartite}. Later, the estimates proved therein were improved in \cite{briet11}, by showing that
\begin{align*}
LV\big(\mathcal{Q}_{cor}^3(N), \mathcal{L}_{cor}^3(N)\big)\geq CN^{\frac{1}{4}},
\end{align*} where $C$ is a universal constant. Moreover, it is known that this estimate is not far from being optimal, since the following inequality holds for every $N$ (see \cite{briet11}):
\begin{align*}
LV\big(\mathcal{Q}_{cor}^3(N), \mathcal{L}_{cor}^3(N)\big)\leq C\sqrt{N},
\end{align*}where $C$ is a universal constant. Note that, throughout this work, $C$ will always denote a general constant, not necessarily the same one each time that it appears. 

It turns out that quantum behaviours cannot lead to an unbounded violation with respect to bilocal behaviours in the correlation scenario, i.e. it cannot hold that  $$\lim_{N\rightarrow \infty}LV\big(\mathcal{Q}_{cor}^3(N), \mathcal{BL}^3_{cor}(N)\big)=\infty.$$ In the following proposition, we show that Tsirelson's result already prevents from having such a behaviour (see item (\ref{lv2}) in the following result). We also analyze the ratio between the set of bilocal correlations and the sets of fully local and quantum correlations.

\begin{proposition}
Given $N$, the following inequalities hold:
\begin{enumerate}
\item \label{lv2} $LV\big(\mathcal{Q}_{cor}^3(N), \mathcal{BL}^3_{cor}(N)\big)\leq K_G$.
\item \label{lv3} $LV\big(\mathcal{BL}^3_{cor}(N), \mathcal{L}_{cor}^3(N)\big)\leq \sqrt{2N}$. This implies $LV\big(\mathcal{BL}^3_{cor}(N), \mathcal{Q}_{cor}^3(N)\big)\leq  \sqrt{2N}$. Moreover, the order $\sqrt{N}$ is optimal in these inequalities, since, in particular, $LV\big(\mathcal{BL}^3_{cor}(N), \mathcal{Q}_{cor}^3(N)\big)\geq \sqrt{N}/(4K_G)$.
\end{enumerate}
\end{proposition}

\begin{proof}

To prove (\ref{lv2}) consider a general Bell functional $M=(M_{xyz})_{x,y,z=1}^N$. Then, for a given tripartite quantum correlation of the form $$\gamma=\Big(\langle\psi|A_x\otimes B_y\otimes C_z|\psi\rangle\Big)_{x,y,z=1}^N,$$ we have
\begin{align*}
\langle M, \gamma\rangle&=\Big|\sum_{x,y,z}M_{xyz}\langle\psi|A_x\otimes B_y\otimes C_z|\psi\rangle\Big|= \Big|\sum_{x,y,z}M_{xyz}\langle\psi|A_x\otimes D_{yz}|\psi\rangle \Big|\\
&\leq K_G \sup_{a_{x}=\pm1,b_{yz}=\pm1}\Big|\sum_{x,y,z}M_{xyz}a_x b_{yz}\Big|\leq K_G \sup_{\delta\in\mathcal{BL}^3_{cor}(N)}|\langle M,\delta\rangle|,
\end{align*}where in the second equality we have denoted $D_{yz}=B_y\otimes C_z$, which is a self adjoint norm one operator, in the first inequality we have used Eq. (\ref{Tsirelosn}) and in the last inequality we have used Proposition \ref{Prop correlation bilocal}.

To prove (\ref{lv3}) consider a generic extremal strategy for the set $\mathcal{BL}^3_{cor}(N)$, which we will assume, without loss of generality, that is of the form $$\gamma=(\gamma_{xy}^{\mathcal{G}}\gamma_z)_{x,y,z=1}^N,$$ where $\gamma^{\mathcal{G}}=(\gamma_{xy}^{\mathcal{G}})_{xy}\in \mathcal{NS}_{cor}^2(N)$. Then,
\begin{align*}
\langle M, \gamma \rangle&=\Big|\sum_{x,y,z}M_{xyz}\gamma_{xy}^{\mathcal{G}}\gamma_z\Big|=\sqrt{2N}\Big|\sum_{x,y,z}M_{xyz}\Big(\frac{\gamma_{xy}^{\mathcal{G}}}{\sqrt{2N}}\Big)\gamma_z\Big|\\
&=\sqrt{2N}\Big|\sum_{x,y,z}M_{xyz}\widetilde{\gamma}_{xy}^{\mathcal{C}}\gamma_z\Big|\leq \sqrt{2N} \sup_{ \mathcal{L}_{cor}^3(N)}|\langle M,\gamma\rangle|.
\end{align*}
Here, we have that $(\widetilde{\gamma}_{xy}^{\mathcal{C}})_{x,y}:=\Big(\frac{\gamma^{\mathcal{G}}_{xy}}{\sqrt{2N}}\Big)_{x,y}\in\mathcal{L}_{cor}^2$ and so $\Big(\widetilde{\gamma}_{xy}^{\mathcal{C}}\gamma_z\Big)_{x,y,z}\in \mathcal{L}_{cor}^3$. This follows from Eq. (\ref{Eq. upper bound NS-CL}). This proves the first inequality in (\ref{lv3}).

The second inequality in (\ref{lv3}) is straighfordward from the first one and the fact that $ \mathcal{L}_{cor}^3(N) \subset \mathcal{Q}_{cor}^3(N)$.

Finally, let us show the optimality of the order $\sqrt{N}$. To this end consider $n$ such that $N/2< 2^n\leq N$ and let $H_{2^n}=(h_{xy})_{x,1=1}^{2^n}$ be a Hadamard matrix, i.e. it has the property $H_{2^n}H_{2^n}^T=2^n\mathbbm{1}$\footnote{The restriction on the dimension to be a power of 2 guarantees that a Hadamard matrix exists (see e. g. \cite{hadamard}).}. Then, define the Bell functional $M=(M_{xyz})$ as $M_{xyz}=h_{xy}$ for $1\leq x,y\leq 2^n$, $z=1$ and $M_{xyz}=0$ otherwise. We will study the values  $\omega_{\mathcal{BL}^3_{cor}(N)}(M)$ and $\omega_{\mathcal{Q}_{cor}^3(N)}(M)$.

The element $\gamma=(\gamma_{xyz})_{xyz}$ defined by $\gamma_{xyz}=M_{xyz}$ for all $x,y,z$ is clearly in $\mathcal{BL}^3_{cor}(N)$ (since $|\gamma_{xyz}|\leq 1$ for every $x$, $y$, $z$). Then,

$$\omega_{\mathcal{BL}^3_{cor}(N)}(M)\geq \Big|\sum_{x,y,z}M_{xyz}\gamma_{xyz}\Big|=\Big|\sum_{x,y=1}^{2^n}h_{xy}^2\Big|=2^{2n}>\frac{N^2}{4}.$$

On the other hand, for every quantum correlation $\gamma=\Big(\langle\psi|A_x\otimes B_y\otimes C_z|\psi\rangle\Big)_{x,y,z=1}^N,$ we have
$$\Big|\sum_{x,y,z}M_{xyz}\gamma_{xyz}\Big|=\Big|\sum_{x,y}h_{xy}\langle\psi|A_x\otimes B_y\otimes C_1|\psi\rangle \Big|=\Big|\sum_{x,y}h_{xy}\langle u_x| v_y\rangle\Big|,$$where we have defined $|u_x\rangle=A_x\otimes\mathbbm{1}\otimes C_{1}|\psi\rangle$ and $|u_y\rangle=\mathbbm{1}\otimes B_y\otimes\mathbbm{1}|\psi\rangle$. We can now apply Eq. (\ref{Tsirelosn}) to upper bound the previous expression by
$$K_G\sup_{a_x,b_y=\pm1}\Big|\sum_{x,y=1}^{2^n}h_{xy}a_xb_y\Big|\leq K_G(2^{n})^{3/2}\leq K_G N^{3/2},$$
where the last inequality is proved in \cite[Ex. 29]{LCPW}.

Since the previous estimate holds for all quantum correlations, the upper bound $\omega_{\mathcal{Q}_{cor}^3(N)}(M)\leq K_G N^{3/2}$ follows. Hence, we deduce $$LV\big(\mathcal{BL}^3_{cor}(N), \mathcal{Q}_{cor}^3(N)\big)\geq \frac{\sqrt{N}}{4K_G}.$$
\end{proof}

\section{General probability distributions}
 
The previous section motivates the question of whether we can obtain unbounded violations of tripartite quantum probability distributions over bilocal probability distributions. As we have seen, this is impossible in the setting of correlations and we want to investigate here if
$$\lim_{\substack{N\rightarrow \infty\\K\rightarrow \infty}}LV\big(\mathcal{Q}^3(N,K), \mathcal{BL}^3_{\mathcal{A}}(N,K)\big)=\infty,\text{ for $\mathcal{A}=\mathcal{NS}$ or $\mathcal{G}$,}$$
holds. Note that, according to the comments after Eq. (\ref{Larg viol}), the previous question is well posed since it is well known that the set of $k$-partite quantum probability distributions is contained in the affine hull of the set of $k$-partite fully local probability distributions. We will show that this is in fact true in the strongest case, $\mathcal{A}=\mathcal{G}$, using $k$-prover one-round games, which are particular Bell functionals. Indeed, one such a game is a Bell functional whose coefficients are of the form 
\begin{align}\label{2P1R-games}
G_{x_1\ldots x_k}^{a_1\ldots a_k}=\pi(x_1,\ldots , x_k)V^{a_1\ldots a_k}_{x_1\ldots x_k},
\end{align}
where $(\pi(x_1\ldots , x_k))_{x_1,\ldots , x_k}$ is a probability distribution and $V$ is a predicate function taking values one or zero. Note that, in particular, $G$ has non-negative coefficients.  Games describe a setting where each of $k$ players  is asked a certain question $x_i$ according to the probability distribution $\pi$ and must answer a certain output $a_i$, being the condition of winning the game that the questions and answers verify $V^{a_1\ldots a_k}_{x_1\ldots x_k}=1$. In this context, the quantity $\omega_{\mathcal{A}}(G)=\sup_{P\in\mathcal{A}}\langle G,P\rangle$ represents the winning probability of the game if the players are restricted to the use of strategies defined by the set $\mathcal{A}$.

Nevertheless, for the purpose of this work, we will consider a slightly more general definition for games and we will treat them as functionals $G$ with non negative coefficients such that 
\begin{align}\label{normalization condition}
\sum_{x_1,\ldots , x_k}\max_{a_1,\ldots, a_k}G_{x_1\ldots x_k}^{a_1\ldots a_k}\leq 1 \hspace{0.4 cm} \text{(normalization condition).}
\end{align}

Let us first recall that, given a bipartite game $G=(G_{xy}^{ab})_{x,y,a,b}$ with $N$ inputs and $K$ outputs per party, we denote by  $G^{\otimes_2}$ the bipartite game with $N^2$ inputs and $K^2$ outputs per party, whose coefficients are: $$G_{x_1x_2y_1 y_2}^{a_1 a_2b_1 b_2}=G_{x_1y_1}^{a_1 b_1}G_{x_2y_2}^{a_2 b_2}.$$ That is, Alice's  and Bob's inputs are $(x_1, x_2)$ and $(y_1, y_2)$, respectively, and Alice's  and Bob's outputs are $(a_1, a_2)$ and $(b_1, b_2)$, respectively. This means that the parties are playing two instances of the game simultaneously. Studying the classical value $G^{\otimes_2}$ for a given game $G$ is the core of the results about parallel repetition theorems, which are of great relevance in computer science.

To make the following result more intuitive, let us explain that our aim is to define a tripartite game $\widetilde{G}$ from a bipartite one $G$. Although it will not be presented here, the first (and somehow easiest) construction we considered was based on two instances of the bipartite game, one for Alice and Bob and the other for Bob and Charlie. In this situation, Alice receives input $x$ and she outputs $a$, Bob receives $(y_1,y_2)$ and outputs $(b_1,b_2)$ (the first is for the game he is playing with Alice and the second, with Charlie) and, finally, Charlie receives input $z$ and outputs $c$. Then the coefficients have the form: 
$$\widetilde{G}_{xy_1y_2z}^{ab_1b_2c}=G_{xy_1}^{ab_1} G_{y_2z}^{b_2c}.$$

In order to find an example such that it does not only give unbounded violations between $\mathcal{Q}^3$ and $\mathcal{BL}_{\mathcal{G}}^3$, but it is also optimal in some parameters, we will present here another construction using three instead of two instances of the game. More precisely, one instance of the game will be asked to Alice and Bob, another to Bob and Charlie and another to Charlie and Alice. Hence, the coefficients of the new game will have the following form: 

$$\hat{G}_{x_1x_2y_1y_2z_1z_2}^{a_1a_2b_1b_2c_1c_2}=G_{x_1y_1}^{a_1b_1}G_{x_2z_1}^{a_2c_1}G_{y_2z_2}^{b_2c_2}.$$

\begin{theorem}\label{thm2}
Let $G$ be a bipartite game with $N$ inputs and $K$ outputs per party. 
Then, the construction of the paragraph above leads to a tripartite game $\hat{G}$ with $N^2$ inputs and $K^2$ outputs per player, such that 
\begin{align}\label{Estimate in them2}
\frac{\omega_{\mathcal{Q}^3}(\hat{G})}{\omega_{\mathcal{BL}_{\mathcal G}^3}(\hat{G})}\geq \frac{\omega_{\mathcal{Q}^2}(G)^3}{\omega_{\mathcal{L}^2}(G^{\otimes_2})}.
\end{align}

Moreover, if $\omega_{\mathcal{Q}^2}(G)$ is attained with local dimension $d$, then Eq. (\ref{Estimate in them2}) is attained with local dimension $d^2$.
\end{theorem}

\begin{proof}
To show that $\omega_{\mathcal{Q}^3}(\hat{G})\geq \omega_{\mathcal{Q}^2}(G)^3$, first note that there must exist a quantum strategy which uses a quantum state $|\phi\rangle$ in some Hilbert space $\mathcal{H}_1\otimes\mathcal{H}_2$ and POVMs $\{\Pi_x^a\}_{a=1}^n$ and $\{\Lambda_y^b\}_{b=1}^n$ acting on the Hilbert spaces $\mathcal{H}_1$ and $\mathcal{H}_2$ in such a way that\footnote{Although the value $\omega_{\mathcal{Q}^2}(G)$ could be not attained, we can find a quantum strategy up to arbitrarily high precision. We avoid writing inequalities up to $\epsilon$.}:
$$\Big|\sum_{x,y,a,b}G_{xy}^{ab}\langle\phi| \Pi_x^a\otimes \Lambda_y^b|\phi\rangle\Big|=\omega_{\mathcal{Q}^2}(G).$$

Then we can consider the tripartite quantum state $|\psi\rangle=|\phi\rangle|\phi\rangle|\phi\rangle\in\big(\mathcal{H}_1\otimes\mathcal{H}_2\big)\otimes\big(\mathcal{H}_1\otimes\mathcal{H}_2\big)\otimes\big(\mathcal{H}_1\otimes\mathcal{H}_2\big)$ and we can define the operators $E_{x_1,x_2}^{a_1,a_2}=\Pi_{x_1}^{a_1}\otimes \Pi_{x_2}^{a_2}$, $F_{y_1,y_2}^{b_1,b_2}=\Lambda_{y_1}^{b_1}\otimes\Pi_{y_2}^{b_2}$ and $G_{z_1,z_2}^{c_1,c_2}=\Lambda_{z_1}^{c_1}\otimes\Lambda_{z_2}^{c_2}$. It is clear that  $\{E_{x_1,x_2}^{a_1,a_2}\}_{a_1,a_2}$, $\{F_{y_1,y_2}^{b_1,b_2}\}_{b_1,b_2}$ and $\{G_{z_1,z_2}^{c_1,c_2}\}_{c_1,c_2}$ are POVMs for all $x_1$, $x_2$, $y_1$, $y_2$, $z_1$, $z_2$. Moreover,
\begin{align*}
\omega_{\mathcal{Q}^3}(\hat{G})&\geq \sum\hat{G}_{x_1x_2y_1y_2z_1z_2}^{a_1a_2b_1b_2c_1c_2}\langle\psi|E_{x_1,x_2}^{a_1,a_2}\otimes F_{y_1,y_2}^{b_1,b_2}\otimes G_{z_1,z_2}^{c_1,c_2}|\psi\rangle \\
&=\Big(\sum_{x_1,y_1,a_1,b_1}G_{x_1y_1}^{a_1b_1}\langle\phi| \Pi_{x_1}^{a_1}\otimes \Lambda_{y_1}^{b_1}|\phi\rangle\Big)\times\Big(\sum_{x_2,z_1,a_2,c_1}G_{x_2z_1}^{a_2c_1}\langle\phi| \Pi_{x_2}^{a_2}\otimes \Lambda_{z_1}^{c_1}|\phi\rangle\Big)\\
&\times\Big(\sum_{y_2,z_2,b_2,c_2}G_{y_2z_2}^{b_2c_2}\langle\phi| \Pi_{y_2}^{b_2}\otimes \Lambda_{z_2}^{c_2}|\phi\rangle\Big)=\omega_{\mathcal{Q}^2}(G)^3,
\end{align*}where the first sum runs over all indices.

Note also that, if we assume $\dim \mathcal{H}_1=\dim \mathcal{H}_2=d$, then, by construction, the local dimension of the quantum state $|\psi\rangle$ is $d^2$.

In order to prove the corresponding upper bound for the classical value, let us consider a bilocal probability distribution $P$ of the form $$\big(P_1(a_1,a_2,b_1,b_2|x_1,x_2,y_1,y_2)P_2(c_1,c_2|z_1,z_2)\big)_{x_1,x_2,y_1,y_2, z_1,z_2}^{a_1,a_2,b_1,b_2, c_1,c_2}$$ and the other two cases will follow by symmetry. 

First of all, notice that, given a certain positive pointwise element $(f(a_2,b_2,x_2,y_2))_{a_2,b_2,x_2,y_2}$ such that $\sum_{a_2,b_2}f(a_2,b_2,x_2,y_2)\leq1$ for all $x_2$ and $y_2$, we can find a probability distribution $\widetilde{P}$ for which all its components are greater than or equal to those of $f$ by defining:

\[ \widetilde{P}(a_2,b_2|x_2,y_2)=\begin{cases} 
      f(a_2,b_2,x_2,y_2) & \text{if } 1\leq a_2,b_2\leq K,\, (a_2,b_2)\neq (K,K),\\
      1-\sum_{(a'_2,b'_2)\neq (K,K)}f(a_2',b_2',x_2,y_2) & \text{if } a_2=b_2=K.
   \end{cases}
\]

Then, using the upper bound for the classical value of $G^{\otimes_2}$, we can write
\begin{align}\label{aux I}
&\sum_{x_2,z_1,a_2,c_1,y_2,z_2,b_2,c_2}G_{x_2z_1}^{a_2c_1}G_{y_2z_2}^{b_2c_2} f(a_2,b_2,x_2,y_2)P(c_1,c_2|z_1,z_2)\\&\nonumber\leq \sum_{x_2,z_1,a_2,c_1,y_2,z_2,b_2,c_2}G_{x_2z_1}^{a_2c_1}G_{y_2z_2}^{b_2c_2} \widetilde{P}(a_2,b_2|x_2,y_2)P(c_1,c_2|z_1,z_2)\leq \omega_{\mathcal{L}^2}(G^{\otimes_2}).
\end{align}

Hence, we have
\begin{align}\label{classical three games}
\langle{\hat{G}}, P\rangle&=\sum G_{x_1y_1}^{a_1b_1}G_{x_2z_1}^{a_2c_1}G_{y_2z_2}^{b_2c_2}P(a_1,a_2,b_1,b_2|x_1,x_2,y_1,y_2)P(c_1,c_2|z_1,z_2)\\\nonumber
&=\sum_{x_2,z_1,a_2,c_1,y_2,z_2,b_2,c_2} G_{x_2z_1}^{a_2c_1}G_{y_2z_2}^{b_2c_2}\Big(\sum_{x_1,y_1,a_1,b_1} G_{x_1y_1}^{a_1b_1}P(a_1,a_2,b_1,b_2|x_1,x_2,y_1,y_2)\Big)P(c_1,c_2|z_1,z_2)\\ \nonumber
&=\sum_{x_2,z_1,a_2,c_1,y_2,z_2,b_2,c_2} G_{x_2z_1}^{a_2c_1}G_{y_2z_2}^{b_2c_2} f(a_2,b_2,x_2,y_2)P(c_1,c_2|z_1,z_2)\leq\omega_{\mathcal{L}^2}(G^{\otimes_2}),
\end{align}where the first sum runs over all indices, we have defined $$f(a_2,b_2,x_2,y_2)=\sum_{x_1,y_1,a_1,b_1} G_{x_1y_1}^{a_1b_1}P(a_1,a_2,b_1,b_2|x_1,x_2,y_1,y_2)$$and the last inequality in Eq. (\ref{classical three games}) follows from Eq. (\ref{aux I}) and the fact that $\sum_{a_2,b_2}f(a_2,b_2,x_2,y_2)\leq1$ for all $x_2$ and $y_2$. To show this last claim, fix $x_2$ and $y_2$, and write
\begin{align*}
&\sum_{a_2,b_2}\sum_{x_1,y_1,a_1,b_1}G_{x_1y_1}^{a_1b_1}P(a_1,a_2,b_1,b_2|x_1,x_2,y_1,y_2)\\
&=\sum_{x_1,y_1,a_1,b_1}G_{x_1y_1}^{a_1b_1}\sum_{a_2, b_2}P(a_1,a_2,b_1,b_2|x_1,x_2,y_1,y_2)\\&
\leq \sum_{x_1, y_1}\max_{a_1,b_1}G_{x_1y_1}^{a_1b_1}\sum_{a_1,a_2,  b_1, b_2,}P(a_1,a_2,b_1,b_2|x_1,x_2,y_1,y_2)\\&=\sum_{x_1, y_1}\max_{a_1,b_1}G_{x_1y_1}^{a_1b_1}\leq 1,
\end{align*}
because of Eq. (\ref{normalization condition}).
\end{proof}

There are two interesting applications of the previous theorem. The first one comes from the application to pseudotelepathy games. That is, those bipartite games which can be won perfectly with quantum strategies but not with classical ones (as it is, for instance, the magic square game \cite{CHTW04}). As a consequence, our construction leads to the existence of pseudotelepathy against bilocality.
\begin{corollary}
Let $G$ be a pseudotelepathy game. Applying the construction of Theorem \ref{thm2} we obtain a tripartite game $\hat{G}$ such that $\omega_{\mathcal{Q}^3}(\hat{G})=1$ and $\omega_{\mathcal{BL}^3_{\mathcal G}}(\hat{G})<1$.
\end{corollary}

The second, and more important, application is to obtain an unbounded violation between tripartite quantum and bilocal conditional probability distributions. For that purpose we will use the \emph{Khot-Vishnoi game}, $G_{KV}$ or KV game \cite{KhVi}, which we briefly explain here. For any $n=2^l$ with $l\in\N$ and every $\eta\in [0,\frac{1}{2}]$ we consider the group $\{0,1\}^n$ and the Hadamard subgroup $H$. Then, we consider the quotient group $G=\{0,1\}^n/H$ which is formed by $\frac{2^n}{n}$ cosets $[x]$ each with $n$ elements. The questions of the games $(x,y)$ are associated to the cosets whereas the answers $a$ and $b$ are indexed in $[n]$. The referee chooses a uniformly random coset $[x]$ and one element $z\in \{0,1\}^n$ according to the probability distribution $pr(z(i)=1)=\eta$, $pr((z(i)=0)=1-\eta$ independently of $i$. Then, the referee asks question $[x]$ to Alice and question $[x\oplus z]$ to Bob. Alice and Bob must answer one element of their corresponding cosets and they win the game if and only if $a\oplus b=z$. Although the KV game is not a two-prover one-round game in the sense of Eq. (\ref{2P1R-games}), it is very easy to see that it verifies the normalization condition given in Eq. (\ref{normalization condition}).

Hence, the Khot-Vishnoi game has $N=2^n/n$ inputs and $K=n$ outputs per player and it can be proved (\cite[Theorem 7]{BuhrmanRSW12}) that
\begin{align}\label{estimate KV}
\omega_{\mathcal{L}^2}(G_{KV})\leq C/n  \hspace{0.2 cm}\text{ and } \hspace{0.2 cm} \omega_{\mathcal{Q}^2}(G_{KV})\geq D/\ln^2 n,
\end{align}
for some universal constants $C$ and $D$. 

The next lemma is necessary in order to apply the Khot-Vishnoi to Theorem \ref{thm2} and it essentially shows that the classical value of the game is multiplicative. 

\begin{lemma}\label{tensor product KV}
Let $G_{KV}$ be the Khot-Vishnoi game. Then, $$\omega_{\mathcal L^2}(G_{KV}^{\otimes_2})\leq C\frac{1}{n^2},$$where $C$ is a universal constant.
\end{lemma}

\begin{proof}
The proof of this result follows exactly the same steps as in the proof of \cite[Theorem 4.1]{BuhrmanRSW12}. As it is explained there, a deterministic strategy (which corresponds to an extremal classical probability distribution) can be identified with a couple of boolean functions $A, B:\{0,1\}^n\rightarrow \{0,1\}$ such that  each of them verifies that, restricted to each coset $[x]$ (see explanation of the game right before this lemma), takes value one for one of the elements and zero for the rest. Then, the winning probability of the game can be written as $$n\mathbb E_u\mathbb E_z [A(u)B(u\oplus z)],$$where $u$ is sampled uniformly at random in $\{0,1\}^n$ and $z\in \{0,1\}^n$ is sampled pointwise independently according to the probability distribution $pr(z(i)=1)=\eta$, $pr((z(i)=0)=1-\eta$. We fix here $\eta=1/2-1/ \log n$. Then, Cauchy-Schwarz inequality followed by a use of the hypercontractive inequality lead to the classical upper bound stated in Eq. (\ref{estimate KV}). 

In the case of $G_{KV}^{\otimes_2}$, a deterministic strategy can be identified with a couple of boolean functions $A, B:\{0,1\}^{2n}=\{0,1\}^n\times \{0,1\}^n\rightarrow \{0,1\}$ such that each of them verifies that, restricted to each pair $[x]\times [y]$, takes value one for one of the elements and zero for the rest. Then, the key point is that sampling $u=(u_1,u_2)$ so that $u_i$ is sampled uniformly in $\{0,1\}^n$ for $i=1,2$ is the same as sampling $u$ uniformly in $\{0,1\}^{2n}$. At the same time,  since $z$ is sampled pointwise independently, we can sample in the form $z=(z_1, z_2)$ where $z_i\in\{0,1\}^n$ is sampled as in the single game for $i=1,2$. Then, the winning probability of the game can be written as $$n^2\mathbb E_{u_1,u_2}\mathbb E_{z_1,z_2} [A(u_1, u_2)B((u_1,u_2)\oplus (z_1, z_2))]=n^2\mathbb E_{u}\mathbb E_{z} [A(u)B(u\oplus z)].$$

Then, doing the same computations as in the proof of \cite[Theorem 4.1]{BuhrmanRSW12} we obtain the bound $n^2\big(\frac{1}{n^2}\big)^{1/(1-\eta)}\leq C/n^2$. This concludes the proof.
\end{proof}

\begin{corollary}\label{thm optimal}
The KV game leads to a tripartite game $\hat{G}$ with $(2^n/n)^2$ inputs and $n^2$ outputs per player,  such that
\begin{align}\label{Eq. estimate optimal}
\frac{\omega_{\mathcal{Q}^3}(\hat{G})}{\omega_{\mathcal{BL}^3}(\hat{G})}\geq C\frac{n^2}{\ln^6 n},
\end{align}and the quantum lower bound in the previous equation is attained with a quantum state of local dimension $n^2$. 

Moreover, this estimate is essentially optimal in the number of outputs and in the local dimension of the Hilbert space.
\end{corollary}

First, this shows that tripartite quantum probability distributions can lead to unbounded violations with respect to bilocal ones. As we have seen in the previous section, this is in hard contrast with the case of correlations. But, second, this also proves optimality in the following sense. Once we have an example in which there exists unbounded violations, a natural question is how far our example is from the best possible construction. That is, does Corollary \ref{thm optimal} provide the best possible violation as a function of the number of inputs and outputs? In fact, when comparing quantum distributions with local (or bilocal) distributions, the amount of violations must be seen as a function of three parameters: number of inputs $N$, number of outputs $K$ and the local dimension of the Hilbert space $d$ used in the quantum distribution. As an example of this, in the bipartite scenario, it can be seen  (\cite{JungeP11low, PV-Survey}) that
\begin{align}\label{upper bounds for Bell inequalities}
LV(\mathcal{Q}^2, \mathcal{L}^2)\leq C\min \{N, K,d\},
\end{align} where $C$ is a universal constant. Interestingly, the KV game provides an example which is essentially optimal (up to logarithmic factors) in the number of outputs and in the dimension $d$, since the estimate in Eq. (\ref{estimate KV}) for $\omega_{\mathcal{Q}^2}(G_{KV})$ is attained by using the maximally entangled state in dimension $n$. It is not known if the upper bound $O(N)$ can be attained, being $\sqrt{N}$ the best lower bound as a function of the number of inputs (see \cite{JungeP11low} for the corresponding game).

Corollary \ref{thm optimal} shows optimality in terms of the number of outputs and dimension of the Hilbert space. It follows from the following two lemmas: 

\begin{lemma}\label{lemma d}
Given a tripartite game $G$. If we denote by $\omega_{\mathcal{Q}_d^3}(G)$ the quantum value of $G$ when at least one of the player is restricted to local dimension $d$, then
$$\omega_{\mathcal{Q}_d^3}(G)\leq d\,  \omega_{\mathcal{BL}_{\mathcal {G}}^3}(G).$$
\end{lemma}

\begin{lemma}\label{lemma k}
Given a tripartite game $G$ with $K$ outputs per player. Then,
$$\omega_{\mathcal{Q}^3}(G)\leq K\omega_{\mathcal{BL}_{\mathcal G}^3}(G).$$
\end{lemma}
\begin{proof}[Proof of Lemma \ref{lemma d}]
The proof can be obtained from a slight modification of  the comments below \cite[Proposition 5.2]{PV-Survey}. Indeed, let us fix a quantum distribution $P$ which is defined by a quantum state $|\psi\rangle\in \C^d\otimes \C^n\otimes \C^m$ and some POVMs $\{\Pi_x^a\}_a$, $\{\Lambda_y^b\}_b$ and $\{\Upsilon_z^c\}_c$ acting on  $\C^d$, $\C^n$ and $\C^m$ respectively, for every $x,y,z$. Then, from the Schmidt decomposition, we deduce that the state $|\psi\rangle$ can be written as $$|\psi\rangle=\sum_{i=1}^d\lambda_i |f_i\rangle |g_i\rangle,$$where $\sum_i|\lambda_i|^2=1$, and $|f_i\rangle$ and $|g_i\rangle$ are orthonormal systems in the unit ball of $\C^d$ and $\C^{nm}$ respectively. Then, we have 
\begin{align*}
|\langle G,P\rangle|&=\Big|\sum_{x,y,z,a,b,c}G_{xyz}^{abc}\langle \psi|\Pi_x^a\otimes \Lambda_y^b\otimes \Upsilon_z^c|\psi\rangle\Big|\\&\leq \sum_{i,j}|\lambda_i||\lambda_{j}|\sum_{x,y,z,a,b,c}G_{xyz}^{abc}|\langle f_i|\Pi_x^a|f_{j}\rangle||\langle g_i |\Lambda_y^b\otimes \Upsilon_z^c|g_{j}\rangle|\\&\leq d \max_{i,j}\sum_{x,y,z,a,b,c}G_{xyz}^{abc}|\langle f_i|\Pi_x^a|f_{j}\rangle||\langle g_i |\Lambda_y^b\otimes \Upsilon_z^c|g_{j}\rangle|,
\end{align*}where we have used the well known fact $\sum_{i=1}^d|\lambda_i|\leq \sqrt{d}\Big(\sum_{i=1}^d|\lambda_i|^2\Big)^{\frac{1}{2}}$.

Now, as it is shown in the comments below \cite[Proposition 5.2]{PV-Survey}, Cauchy-Schwarz inequality implies that for every $i$ and $j$, and for every $x$, $y$ and $z$, we have
\begin{align*}
\sum_a|\langle f_i|\Pi_x^a|f_{j}\rangle|\leq 1 \hspace{0.2 cm} \text{   and  } \hspace{0.2 cm}  \sum_{b,c}|\langle g_i |\Lambda_y^b\otimes \Upsilon_z^c|g_{j}\rangle|\leq 1.
\end{align*}Hence, we deduce that $\langle G,P\rangle\leq d\, \omega_{\mathcal{BL}_{\mathcal {G}}^3}(G)$, which concludes the proof.
\end{proof}

\begin{proof}[Proof of Lemma \ref{lemma k}]
Notice that given a tripartite game $G=(G_{xyz}^{abc})_{x,y,z,a,b,c}$, we can define three different bipartite games $G_1$, $G_2$ and $G_3$ in which one party receives two inputs and gives two outputs. Put differently, we join Alice and Bob for $G_1$, Bob and Charlie for $G_2$ and Charlie and Alice for $G_3$. In case that the maximum value for $\omega_{\mathcal{BL}_{\mathcal {G}}^3}(G)$ is attained for a bilocal probability distribution of the form $(P(a,b|x,y)P(c|z))_{xyz}^{abc}$, then this probability distribution can be seen as a local bipartite probability distribution in the scenario where Alice and Bob are joint, and it will give the maximum value for the bipartite game $G_1$. As the other cases are similar we can say that:
\begin{align}\label{sub functionals}
\omega_{\mathcal{BL}_{\mathcal {G}}^3}(G)=\max\{\omega_{\mathcal{L}^2}(G_1),\omega_{\mathcal{L}^2}(G_2),\omega_{\mathcal{L}^2}(G_3)\}.
\end{align}

Then, the upper bound we want to show follows from the known estimate for bipartite games (see comments below \cite[Proposition 4.5]{PV-Survey}) $$\omega_{\mathcal{Q}^2}(G)\leq\min\{K_1,K_2\}\omega_{\mathcal{L}^2}(G).$$
Indeed, according to our hypothesis, the functionals $G_1$, $G_2$ and $G_3$ from $G$ have $K$ outputs for one player and $K^2$ outputs for the other player. Then, for every $i=1,2,3$, we have
$$\omega_{\mathcal{Q}^3}(G)\leq \omega_{\mathcal{Q}^2}(G_i)\leq\min\{K,K^2\}\omega_{\mathcal{L}^2}(G_i)=K\omega_{\mathcal{L}^2}(G_i),$$ which, according to Eq. (\ref{sub functionals}) gives the desired upper bound.
\end{proof}

\section{Conclusions}

In this work we have extended the study of relative Bell violations of quantum resources over local and fully local ones to the genuinely multipartite scenario by comparing the power of quantum strategies over bilocal models. We have considered first the correlation scenario, where we have found that, as in the bipartite case, the ratio of Bell violation of quantum behaviours over bilocal ones is upper bounded by Grothendieck's constant for any number of inputs and, hence, there cannot be unbounded Bell violations. Since not all bilocal correlations are reproducible by quantum models, we have also investigated the relative power of the former over the latter. We have shown that this ratio is upper bounded by $O(\sqrt{N})$ and that this order is optimal. Next, we have considered the case of general conditional probability distributions. Contrary to the previous case, we have obtained here that quantum strategies lead to unbounded Bell violations over general bilocal behaviours. In order to do so, we have proved that if one considers tensor products of bipartite games the ratio of the quantum value over the general bilocal value is related to the ratio of the quantum and local values for the bipartite game. This has allowed us, by considering explicit games such as the Khot-Vishnoi game, to establish that there exist games for which the ratio of the quantum and general bilocal values grows unboundedly with the number of inputs and outputs. We moreover have proven that for a particular choice of games the given estimate of the asymptotic behaviour of this ratio is essentially optimal in the number of outputs and in the dimension of the Hilbert space. 

It might be worth mentioning that the above games require a large number of inputs, i.e.\ exponential in the obtained violation. Random constructions of Bell functionals (see \cite{Pala2015, PV-Survey} for some surveys on this topic) could be used to show that there exist Bell inequalities in which the number of inputs and outputs grow polynomially with the amount of violation. This, however, is a non-constructive procedure and would come at the expense of not identifying an explicit Bell functional for this task.

It should be noticed that the two results about the ratio of Bell violations of quantum behaviours over bilocal ones -- boundedness in the correlation setting and unboundedness in the general case -- hold irrespectively of whether we consider general bilocal or non-signalling bilocal models. In the first case, this holds because both sets of models happen to coincide in the correlation scenario, as discussed in Sec.\ \ref{Sec: Bi-local correlations}. In the second case, unboudedness with respect to general bilocal behaviours automatically implies the same with respect to non-signalling bilocal behaviours due to the fact that this latter set is included in the former. Thus, our result can also be understood as showing an unlimited advantage of GMNL quantum behaviours irrespective of the underlying definition of bilocality. As mentioned in the introduction, the correlations contained in general bilocal models might be so strong that lead to undesirable unphysical effects in certain scenarios and this has motivated to consider more constrained hybrid models. Despite this fact, not only general bilocal models are unable to simulate all quantum behaviours as proven by Svetlichny in \cite{sve}, but our results show that quantum-mechanical resources can be, in a certain sense, unboundedly better than this strongest form of bilocality.

\section*{acknowledgment}
This research was funded by the Spanish MINECO through Grant No. MTM2017-88385-P, MTM2017-84098-P, MTM2014-54240-P and by the Comunidad de Madrid through grant QUITEMAD-CM P2018/TCS4342. We also acknowledge funding from SEV-2015-0554-16-3.

\end{document}